\documentclass[12pt,onecolumn,twoside]{IEEEtran}

\usepackage[utf8]{inputenc}
\usepackage[T1]{fontenc}
\usepackage{url}              
\usepackage{cite}             

\usepackage[cmex10]{amsmath}  
\usepackage{mleftright}       
\mleftright                   

\usepackage{graphicx}         
\usepackage{booktabs}         

\usepackage{amsmath}
\usepackage{latexsym}
\usepackage{amssymb}

\newcommand{\cC}{{\cal C}}

\newcommand{\cH}{{\cal H}}

\newcommand{\cS}{{\cal S}}

\newcommand{\bfzero}{{\bf 0}}

\newcommand{\abs}[1]{\left|#1\right|}
\newcommand*{\bigO}{\mathcal{O}}

\newcommand{\field}[1]{\mathbb{#1}}

\newcommand{\F}{\field{F}}

\newcommand{\bldc}{{\mathbf{c}}}

\newcommand{\bldx}{{\mathbf{x}}}

\newtheorem{theorem}{Theorem}
\newtheorem{definition}{Definition}
\newtheorem{lemma}{Lemma}
\newtheorem{corollary}{Corollary}
\newtheorem{example}{Example}

\begin{document}

\title{Covering Sequences\\ and Covering-Sequences Codes}

\author{%
  \IEEEauthorblockN{Tuvi Etzion}
}

\maketitle

\begin{abstract}
An $(n,R)_q$-covering sequence is a cyclic sequence, over the finite field $\F_q$,
whose consecutive $n$-tuples form a code of length $n$ and covering radius $R$.
An $(n,m,R)_q$-covering-sequences code is a set of cyclic sequences of length $m$, over $\F_q$, whose consecutive $n$-tuples form a code of length $n$
and covering radius $R$. These codes are the best building blocks for $(n,R)_q$-covering sequences.
We show, for small radii, how cyclic codes and constacyclic codes with small covering radius,
can be used to construct such sequences of short length and such codes
with a relatively small number of sequences and a total number of codewords in the associated covering code.
Sequences with small radius whose length approaches asymptotically to optimality are constructed, especially
for an alphabet of prime power size large enough. With the same construction, interesting codes are also constructed for larger radii.
\end{abstract}

\section{Introduction}
\label{sec:Intro}

An {\bf \emph{$(n,R)_q$-covering code}} $\cC$ is a set of words of length $n$ over a given alphabet $\Sigma_q$, of size $q$, such that each word
of length $n$ over $\Sigma_q$ is within Hamming distance $R$ from at least one codeword in $\cC$.
In other words, for each $x \in \Sigma_q^n$, there exists $c \in \cC$ such that $d(x,c) \leq R$,
where $d(y,z)$, $y,z \in \Sigma_q^n$, denotes the Hamming distance between $y$ and $z$.
The {\bf \emph{covering radius}} of a code $\cC$ of length $n$, over $\Sigma_q$, is the smallest $R$ such that for
each $x \in \Sigma_q^n$, there exists $c \in \cC$ such that $d(x,c) \leq R$.
Covering codes were always of interest, but the interest increased due to the following three seminal papers~\cite{CKMS,CLS86,GrSl85}.
The interest was also increased partially because of the connection of covering codes to data compression
(see for example~\cite{CHLL97,EWZ95}).
An excellent book that covers all aspects of such codes is~\cite{CHLL97}. One of the main goals is to find for
given $n$, $R$, and $q$, the $(n,R)_q$-covering code of the smallest size. Lot of work was done in this direction,
see the excellent book~\cite{CHLL97} and references therein.

An {\bf \emph{$(n,R)_q$-covering sequence}} (an $(n,R)_q$-CS for short) is a cyclic sequence whose consecutive $n$-tuple form an $(n,R)_q$-covering code.
The target is to find for given $n$, $R$, and $q$, the $(n,R)_q$-CS of the shortest length.
These sequences were considered first by Chung and Cooper~\cite{ChCo04}
who called such a structure a de Bruijn covering code.
The reason for the name was that the $n$-tuples of a cyclic sequence are considered,
and this sequence forms a cycle in the de Bruijn graph. Moreover, if $R=0$ then the shortest such sequence
is a de Bruijn sequence~\cite{deB46,Etz24,Fre82}.

The simple lower bound on the size of $\cC$ is
$$
\abs{\cC} \geq \frac{q^n}{V_q (n,R)},
$$
where
$$
V_q(n,R) = \sum_{i=0}^{R} \binom{n}{i} (q-1)^i ~.
$$
This bound is the {\bf \emph{sphere-covering bound}}.
There exists a covering code that approaches this bound up to a factor roughly $eR \log R$~\cite{KSV03}.
The proof method for this bound is probabilistic.
A similar bound for an $(n,R)_q$-CS over a prime power alphabet was presented in~\cite{ChCo04}. This bound was generalized to any alphabet by
Vu~\cite{vu05}. The bound states that for fixed $R$
there exists an $(n,R)_q$-CS, over $\Sigma_q$, whose length is at most $\bigO\left(\frac{q^n}{V_q(n,R)}\log n \right)$.

For small $R$, there are some $(n,R)_q$-covering codes that attain the sphere-covering bound with equality, such as the Hamming codes
of length $n=2^r-1$ and covering radius one. there are other codes whose size
is very close to the upper bound, such as the two Golay codes, the ones for $R=2$ which are perfect
asymptotically~\cite[pp. 172--174]{Etz22},\cite{EtMo05},\cite[Construction 4.24]{Str94},
or other similar codes~\cite{EtGr93,EtMo05}. Similarly, such sparse covering codes were also considered for $R=3$~\cite{EtGr93,EtMo05}.
The main goal of the research on $(n,R)_q$-CSs is to obtain sequences whose length is
as close as possible to the upper bounds obtained for covering codes. For small $n$, lower and upper bounds
on the sizes of $(n,R)_2$-CSs were obtained in~\cite{ChCo04}. Several constructions that yield upper bounds
on the size of such binary sequences for small and large $n$ were obtained in~\cite{CETV24,CETV25}. Using heuristic search, some bounds
for small $n$ and $R$ were found by~\cite{Ros25}. Two of the constructions presented in~\cite{CETV24,CETV25} yield sequences
whose length is within a small constant factor of optimality. These sequences were obtained only for length $n=2^r$ and $n=2^r-1$
and the radius of the code is only $1$.

Finally, the exposition in~\cite{CETV25} suggested a new type of covering codes.
An {\bf \emph{$(n,m,R)_q$-CS code}} (an $(n,m,R)_q$-CSC for short) $\cC$ is a set of cyclic sequences,
which are the codewords, of length $m$
such that each word of length $n$ is within distance $R$ from at least one $n$-tuple of a codeword of $\cC$, i.e., the
consecutive $n$-tuples of all the codewords of $\cC$ form an $(n,R)_q$-covering code.
Any cyclic code $\cC$ of length~$n$ and covering radius $R$ can be used as an $(n,n,R)_q$-CSC $\cC'$,
where from the codewords that have the same cyclic shifts (which will be referred to later as
an {\bf \emph{orbit}}, when only one representative is taken). An $(n,R)_q$-CS of length $\ell$ is an $(n,\ell,R)_q$-CSC
and hence the covering-sequences codes are the link between cyclic covering codes and covering sequences.
This also can be the measure to evaluate the efficiency of a covering-sequences code. A code that yields a shorter covering sequence
is a better code.

There are four goals for this paper. The first one is to have some exposition for $(n,R)_q$-CSs
and $(n,m,R)_q$-CSCs for $q>2$ as~\cite{CETV25} considered only binary sequences.
The second is to find new $(n,R)_q$-CSs that are within a small constant factor from optimality.
The third is to put a strong emphasis on $(n,m,R)_q$-CSCs and to find new $(n,m,R)_q$-CSCs, especially codes for which $m> n$.
We believe that this should be a main focus in future research on covering codes.
where the total length of the codewords is within a constant factor from optimality.
The fourth goal is to look on the structure of orbits associated with the codes generated in the current paper.

The rest of the paper is organized as follows. Section~\ref{sec:cyclic} describes the main construction
of a $(n,R)_q$-CS from an $(n,m,R)_q$-CSC.
Section~\ref{sec:binary} reviews the constructions of the sequences for $n=2^r$ and $n=2^r-1$. It also emphasizes the
connection of these constructions and $(n,m,R)_q$-CSCs. Self-dual sequences play an important and surprising role in
the construction of $n=2^r$. They can present a nearly perfect covering code as a negacyclic code. All of these concepts
will be discussed in this section. Section~\ref{sec:nonbinary} shifts the discussion towards non-binary sequences
over any finite field $\F_q$. The optimal constructions in the binary case are generalized for $\F_q$.
In particular, constacyclic codes take the role of self-dual sequences in one of the constructions.
The sequences obtained in the constructions are analyzed, and it is shown that the lengths of the obtained sequences
are within a small constant factor of optimality. In particular, a factor of $\frac{q}{q-1}$ from optimality is
obtained for sequences over $\F_q$.
Section~\ref{sec:interleaving} is devoted to a construction
of $(n,R)_q$-CSs and $(n,m,R)_q$-CSCs from codes obtained by interleaving. Interleaving has been presented in the past for
constructions of $(n,R)_q$-CSs, but in this section, interleaving is not performed directly on the sequences, but
on the parity-check matrices of codes used to construct the $(n,R)_q$-CSs. The sequences have a larger radius than
the ones from which they were interleaved.
Conclusions and a list of open problems are suggested in Section~\ref{sec:conclude}.

\section{A Construction from Cyclic CSC}
\label{sec:cyclic}

This section is devoted for the basic idea of the main construction presented in this work and implemented later
for various parameters on various codes. Although the construction in~\cite{CETV25} was defined for a binary alphabet,
generalization for the construction to a non-binary alphabet is straightforward. The construction starts with an $(n,m,R)_q$-CSC $\cC$
with $M$ codewords of length $m$. A sequence $\cS$ is {\bf \emph{degenerated}} if it can be represented as $\cS = [X,X,\ldots,X]$, where the length
of $X$ is smaller than the length of $\cS$. If $X$
is the shortest string in such representation, then the length of $X$ is the {\bf \emph{period}} of the sequence.
The period of $X$ is always a divisor of $m$. If the shortest such string has length~$m$, then the sequence
has {\bf \emph{full-period}}. If two sequences $\cS_1$ and $\cS_2$ are cyclic shifts of each other, then the sequences are
called {\bf \emph{equivalent}} and denoted by $\cS_1 \simeq \cS_2$.
For each codeword $\bldc$ (a cyclic sequence), we have to choose a starting point and let such a codeword
$\bldc=(c_1,c_2,\ldots,c_m)$ be extended to $\bldc'=(c_1,c_2,\ldots,c_m,c_1,c_2,\ldots,c_{n-1})$.
If $\bldc=(c_1,c_2,\ldots,c_m)$ is a degenerated codeword of period $\ell$ that divides $m$,
then its extended codeword is $\bldc'=(c_1,c_2,\ldots,c_\ell,c_1,c_2,\ldots,c_{n-1})$. The extended codewords should be ordered
in a list $\bldc_0'$, $\bldc_1'$, $\bldc_2'$,\ldots,$\bldc_{M-1}'$, in a way that $\bldc_i'$ and $\bldc_{i+1}'$, for $0 \leq i \leq M-1$,
have a common string $X_i$ of length $t_i$, which is suffix of $\bldc_i'$ and a prefix of $\bldc_{i+1}'$, $0 \leq i \leq M-1$, where indices
are taken modulo $M$. In other words, $\bldc_i' = (X_{i-1}, Z_i)=(Y_i X_i)$ and $\bldc_{i+1}'=(X_i,Z_{i+1})=(Y_{i+1},X_{i+1})$.
Now, we can concatenate all the codewords in the list $\bldc_0'$, $\bldc_1'$, $\bldc_2'$,\ldots,$\bldc_{M-1}'$ omitting
the shared parts $X_0$, $X_1$,\ldots,$X_{M-1}$ to one cyclic sequence as follows:
$$
\cS= [Y_0,Y_1,Y_2,\ldots,Y_{M-1}]~.
$$

\begin{lemma}
\label{lem:equivOmit}
The cyclic sequence $\cS' = [ Z_0,Z_1,\ldots,Z_{M-1}]$ is equivalent to the sequence $\cS$.
\end{lemma}
\begin{IEEEproof}
Consider the following sequence, written in three different ways by the definition of $\bldc_i$, $0 \leq i \leq M-1$,
$$
\bldc_0' \bldc_1' \bldc_2' \bldc_3' \cdots \bldc_{M-1}' =[Y_0, X_0, Y_1, X_1, Y_2,\ldots, Y_{M-1},X_{M-1}]
= [X_{M-1},Z_0,X_0,Z_1,X_1,\ldots,X_{M-2},Z_{M-1}].
$$
Omitting the one occurrence of $X_i$, for each $0 \leq i \leq M-1$, from this sequence implies that
$$
\cS= [Y_0,Y_1,Y_2,\ldots,Y_{M-1}] \simeq [Z_0,Z_1,Z_2,\ldots,Z_{M-1}] .
$$
\end{IEEEproof}

\begin{theorem}
The sequence $\cS$ is an $(n,R)_q$-CS.
\end{theorem}
\begin{IEEEproof}
Each word of length $n$ is within distance $R$ from at least one $n$-tuple of a codeword of $\cC$.
Hence, it is sufficient to show that all the $n$-tuples of $\cC$ are also $n$-tuples of $S$.
For an $n$-tuple in the codeword $\bldc_i =(c_1,c_2,\ldots,c_m)$, where $c_1$ is the starting point
of the orbit associated with $\bldc_i$, the extended codeword associated with $\bldc_i$ is $\bldc_i'=(X_{i-1}, Z_i)=(Y_i,X_i)$.
Since $X_i$ appeare in the extended codewords of $c_i$ and the extended codeword of $c_{i+1}$, respectively, i.e.,
$c_i'$ and $c_{i+1}'$, respectively,
the claim can be verified now from Lemma~\ref{lem:equivOmit}.
\end{IEEEproof}

The construction uses merging of cycles and hence it will be referred to as Construction MC (merging cycles).

\section{Binary Covering-Sequences Codes}
\label{sec:binary}

After the description of Construction MC based on cyclic covering-sequences codes, we are going to describe
implementations of the construction. Two such implementations were considered in~\cite{CETV25}.
The first one is based on the cyclic Hamming code, and the second one is based on self-dual sequences
that form a nearly perfect covering code. The review of these two implementations will help to explain
the constructions in Sections~\ref{sec:nonbinary} and~\ref{sec:interleaving} as well as to shed some
new light on the implementation with self-dual sequences.

\subsection{$(2^r-1,1)$-CS based on the Hamming code}
\label{sec:Hamming}

Let $[n,k]_q$ denote a linear code of length $n$ and dimension $k$ over $\F_q$.
The $[2^r-1,2^r-1-r]_2$ Hamming code, $\cH_2(r)$, is a cyclic code, with radius $1$, whose parity-check matrix can be represented as
$$
[ \alpha^0 ~ \alpha^1 ~ \alpha^2 ~ \cdots ~ \alpha^{2^r-2} ],
$$
where $\alpha$ is a primitive element in $\F_{2^r}$.
It is well-known that the number of degenerated orbits in the code is much less than $2^{2^{k-1}}$~\cite[p. 105, Lemma 3.20]{CETV25,Etz24}.
Therefore, for the construction of an $(2^r-1,1)_2$-CS we are using less than $2^{2^r -2r-1} + 2^{2^{k-1}}$ orbits (full-period and degenerated)
that yield a $(2^r-1,1)_2$-CS whose length is less than $2^{2^r -r} + 2^{2^{r-1} +r+1}$. Thus, the sequence is within a factor of less than
$2 + \frac{1}{2^{2^r-2k-2}}$ from optimality.

It is interesting to note that all the binary degenerated sequences of length $2^r-1$ are codewords in the Hamming code and hence
they had to be taken into account in the computations.

\begin{lemma}
\label{lem:BnoFull}
All degenerated words of length $n=2^r-1$ are codewords in the Hamming code of length~$n$.
\end{lemma}
\begin{IEEEproof}
Let $\bldx=(x_0,x_1,\ldots,x_{n-1})$ be a sequence of period $\pi$, where $0 < \pi < n$,
i.e., $\pi$ divides $n$. Note that $\alpha^{\frac{n}{\pi}\pi} = \alpha^n = \alpha^0$ and hence
$$
\left( \sum_{i=0}^{\pi-1} x_i \alpha^i \right) \left( \sum_{j=0}^{\frac{n}{\pi} -1} \alpha^{j\pi} \right) \alpha^\pi =
\left( \sum_{i=0}^{\pi-1} x_i \alpha^i \right) \left( \sum_{j=1}^{\frac{n}{\pi}} \alpha^{j\pi} \right) =
\left( \sum_{i=0}^{\pi-1} x_i \alpha^i \right) \left( \sum_{j=0}^{\frac{n}{\pi} -1} \alpha^{j\pi} \right) .
$$
Since $0 < r < n$, it follows that $\alpha^r \neq 1$ and hence
$$
\left( \sum_{i=0}^{\pi-1} x_i \alpha^i \right) \left( \sum_{j=0}^{\frac{n}{\pi} -1} \alpha^{j\pi} \right) = \bfzero .
$$
Therefore, $\bldx$ is a codeword and all the degenerated necklaces are codewords.
\end{IEEEproof}

\subsection{$(2^r,1)$-CS based on self-dual sequences}
\label{sec:self-dual}

The implementation of Construction MC using self-dual sequences is of a special interest for a few different reasons.
A binary cyclic sequence $\cS = [s_0,s_1,\ldots,s_{k-1}]$ is called {\bf \emph{self-dual sequence}} if it is invariant under completion.
We will refer only to self-dual sequences with full-period, i.e., not degenerated. Such a sequence $S$ can be
represented as $\cS=[ X , \bar{X}]$, where $X$ is a sequence whose length is $k/2$ and $\bar{X}$ is the binary complement of $X$.
In~\cite{BER25,CETV25} there is a construction for a set of $2^{2^r -2r-1}$ such sequences of length $2^{r+1}$
that form a $(2^r,2^{r+1},1)_2$-CSC of size $2^{2^r -2r-1}$. These sequences can be paired in a way that each pair can be merged in a way that
the $2^r$-tuples of each sequence are also $2^r$-tuples of the merged sequence. Merging all these pairs yield
a $(2^r,2^{r+2},1)_2$-CSC $\cC$ of size $2^{2^r -2r-2}$. By applying Construction MC on the sequences of $\cC$ we obtain
a $(2^r,1)_2$-CS of length smaller than $2^{2^r-2k-2}(2^{r+2}+2^r-1)$. An $(2^r,1)_2$-covering code of the smallest size
has $2^{2^r-r}$ codewords. Thus, the $(2^r,1)_2$-CS obtained by Construction MC on the code $\cC$ has size within factor of $1.25$
from optimality.

Moreover, the $n$-tuples of the CSC $\cC$ form what is called a balanced nearly perfect 1-covering code~\cite{BER25}.
This specific code has some more interesting properties~\cite{BER25} that can be also obtained from a union
of an extended Hamming code and its coset. We will refer to this code again in the next section.

\section{Optimal Non-Binary Covering-Sequences Codes}
\label{sec:nonbinary}

The idea behind the two constructions of the $(n,1)_2$-CSs in Section~\ref{sec:binary} is to use concatenation of codewords from cyclic codes
or, more precisely $(n,m,R)_q$-CSCs. In Section~\ref{sec:binary} such binary codes were considered.
In Section~\ref{sec:Hamming} the used code is $\cH_2(r)$, i.e., $n=m$, while in Section~\ref{sec:self-dual}
the code used consists of a set of self-dual sequences and $m >n$. In the non-binary case there are two types
of $(n,m,R)_q$-CSCs too, one for which $m=n$ and one for which $m>n$. But, in both cases $\cH_q(r)$
is used. As in the binary case cyclic codes are used, but instead of self-dual sequences
another type of code is used.

\begin{definition}
\label{def:cyclic}
$~$
\begin{enumerate}
\item A code $\cC$ of length $n$ over $\F_q$ is called a \emph{negacyclic} code if $(c_1,c_2,\ldots,c_n) \in \cC$ implies
that $(c_2,\ldots,c_n,-c_1) \in \cC$.

\item A code $\cC$ of length $n$ over $\F_q$ is called a \emph{constacyclic code} if $(c_1,c_2,\ldots,c_n) \in \cC$
implies that $(c_2,\ldots,c_n,\lambda c_1) \in \cC$, for some given $\lambda \in \F_q$.
\end{enumerate}
\end{definition}

Cyclic codes have been extensively studied in the literature and no reference
is required for this celebrated fact. The same is true for negacyclic codes, e.g.~\cite{ChWu24,KaZh12} and
constacyclic codes, e.g.~\cite{CFLL,CLZ15}.
In some wide sense the code obtained from self-dual sequences, a $(2^r,2^{r+1},1)_2$-CSC
of size $2^{2^r -2r-1}$, can be said to be negacyclic if the alphabet $\{ 0,1 \}$
will be changed to $\{ -1 , +1\}$. For non-binary sequences and codes, $\cH_q(r)$ will be represented as a constacyclic code.
More interesting and also more efficient is the $(2^r,2^{r+2},1)_2$-CSC of size $2^{2^r -2r-2}$. This code is the one used in Construction MC
to generate the shortest $(n,1)_2$-CS. This code is optimal since the number of $2^r$-tuples in the code
is the same as the number of codewords in the smallest $(2^r,1)$-covering code. However, it is the only code we are using
that it is not a cyclic code or can be represented as a constacyclic code.

For a Hamming code over $\F_q$, $q$ a prime power, with redundancy $r$, $\cH_q(r)$, the parity-check matrix is represented by
$n=\frac{q^r-1}{q-1}$ linearly independent column vectors of length $r$. These $n$ columns can be chosen
in a few different ways. For example, we can take all column vectors of length $r$ whose first nonzero entry is an \emph{one}.
We will choose a different representation which resembles the choice for $\cH_2 (r)$. Let $\alpha$~be a primitive element in $\F_{q^r}$
and let $H$ the $\frac{q^r-1}{q-1} \times r$ matrix whose $i$-th column is the vector representing $\alpha^i$,
$0 \leq i < n$.

While $\cH_2(r)$ is always a cyclic code with radius $1$,
$\cH_q(r)$ is a cyclic code if and only if ${\gcd (r,q-1)=1}$\cite[pp. 169--170]{HuPl03}.\cite[p. 244, pp. 253--254]{Rot06}.
This implies that Construction MC can be applied to $\cH_q(r)$ when $\gcd (r,q-1)=1$ and yields similar results,
while, Lemma~\ref{lem:BnoFull} does not have a straightforward generalization for $\F_q$.

We do not have a construction with self-dual sequences for the non-binary case.
Moreover, these sequences form a topic for further research~\cite{Etz26,Etz27}.
Instead, we use the fact that the Hamming code over $\F_q$ can be always
represented as a constacyclic code as follows.
The element $\gamma = \alpha^n$ is a primitive element in $\F_q$.
The parity-check matrix of $\cH_q(r)$ is represented by
$$
[ h_0 ~ h_1 ~ \cdots ~ h_{n-1} ],
$$
where $h_i$ is the $q$-ary representation of $\alpha^i$.

If $(c_0,c_1,\ldots,c_{n-2},c_{n-1})$ is a codeword then $\sum_{i=0}^{n-1} c_i \alpha^i = \bfzero$.
This implies that $\sum_{i=0}^{n-1} c_i \alpha^{i+1} = \bfzero$ or $(c_{n-1} \gamma) \alpha^0 + \sum_{i=1}^{n-1} c_{i-1} \alpha^i =\bfzero$.
As a consequence we have that $(\gamma a_{n-1},c_0,c_1,\ldots,c_{n-2})$ is a codeword in $\cH_q(r)$.
With the same process we have that $(\gamma a_{n-2},\gamma a_{n-1},c_0,c_1,\ldots,c_{n-3})$ is a codeword and so on, so
$(\gamma c_0, \gamma c_1,\ldots, \gamma c_{n-2}, \gamma c_{n-1})$ is a codeword,
$(\gamma^2 c_0, \gamma^2 c_1,\ldots, \gamma^2 c_{n-2}, \gamma^2 c_{n-1})$ is a codeword, and finally
$(\gamma^{q-2} c_0, \gamma^{q-2} c_1,\ldots, \gamma^{q-2} c_{t-2}, \gamma^{q-2} c_{n-1})$ is a codeword. In this process we have that
$$
(c_1,\ldots, c_{n-2}, c_{n-1},\gamma^{q-2} c_0)=(c_1,\ldots, c_{n-2}, c_{n-1},\gamma^{-1} c_0)
$$
is a codeword (note that $\gamma^{-1}$ is also a primitive element in $\F_q$). Hence we have the following results,
which can now be easily verified.
\begin{lemma}
In this representation of $\cH_q(r)$, the period of the cycle generated by a codeword is a divisor of $(q-1)n=q^r-1$.
\end{lemma}
\begin{corollary}
The cycles of the codewords of $\cH_q(r)$ represented as a constacyclic code form an $(n=\frac{q^r-1}{q-1},(q-1)n=q^r-1,1)_q$-CSC.
\end{corollary}

Not all degenerated words of length $(q-1)n$ over $\F_q$ are codewords since not all the required conditions
in the proof of Lemma~\ref{lem:BnoFull} are satisfied.
Of special interest are those codes whose length of codeword of $\cH_q(r)$ is a repunit prime.
A prime of the form $n=\frac{q^r-1}{q-1}$ is called a {\bf \emph{repunit prime}}. When $q=2$ this
prime is a well-known Mersenne prime. If $n$ is a prime, then all the codewords of the related
$(n,q^r -1,1)_q$-CSC are on orbits of full-period except for $q$ codewords, the all-zero codeword
and the $q-1$ codewords whose orbit is $[c_0,c_1,\ldots,c_{q-1}]$, where $c_i = \gamma^i$ and
$\gamma$ is a primitive element in $\F_q$.
Since all orbit except for these two are of full-period it follows that the size of the code is
$$
\frac{q^{n-r} -q}{n(q-1)} +2 = q \frac{q^{n-r-1}-1}{q^r -1} +2 .
$$
Each orbit is extended with $n-1$ bits and using concatenation, we find that the length of the associated $(n,1)_q$-CS is not more than
$$
q \frac{q^{n-r-1}-1}{q^r -1} (q^r -1 +n-1) +q+2(n-1) \leq q^{n-r} -q +(q^{n-r-1}-1) \frac{q}{q-1} \leq q^{n-r} \frac{q}{q-1}
$$
and since the size of optimal $(n,1)_q$-covering code is $q^{n-r}$, it follows that the constructed CSC and CS are within a factor
of $\frac{q}{q-1}$ from optimality. Slightly more complicated analysis achieves the same result when $n$ is not a prime.

The two extreme lengths of the CSCs obtained from $\cH_q(r)$ are of length $\frac{q^r-1}{q-1}$ when ${\gcd (r,q-1)=1}$
and of length $q^r-1$. However, if instead of a primitive element $\gamma$ in $\F_q$, an element of smaller order than $q-1$ is
used and the parity-check matrix is also changed
(where consecutive columns are of the form $\alpha^0 ~ \alpha^i ~ \alpha^{2i} ~ \alpha^{3i} ~ \cdots$),
we can have CSCs whose length is $(2^r-1)/t$, where $t$ divides $2^r-1$.
These codes will be discussed in the full-version of the paper.

\section{Interleaving of parity-Check matrices}
\label{sec:interleaving}

Interleaving two covering sequences, in an appropriate way, yields a new covering sequence with a larger radius and relatively
short length~\cite{CETV25}. In this section, we show that interleaving the parity-check matrix of a cyclic covering code
with the same parity-check matrix yields a better covering sequence whose length can be within a smaller factor from optimality.
The following lemma can be verified from the definition of a cyclic linear covering code.

\begin{lemma}
\label{lem:interP}
Let $H = [ h_0 ~ h_1 ~ h_2 ~ \cdots ~ h_{n-1}]$ be the parity-check matrix of an $[n,k]_q$ linear $(n,R)_q$-covering code $\cC$.
The parity-check matrix
$$
H'=\left[
\begin{array}{ccccccccc}
h_0 & 0 & h_1 & 0 & h_2 & \cdots & 0 & h_{n-1} & 0 \\
0 & h_0 & 0 & h_1 & 0 &  \cdots  & h_{n-2} & 0 & h_{n-1} \\
\end{array}
\right]
$$
is a parity-check matrix of a $[2n,2k]_q$ linear $(2n,2R)_q$-covering code $\cC'$.
Furthermore, if $\cC$ is cyclic (constacyclic, respectively) code, then also $\cC'$ is a cyclic (constacyclic, respectively) code.
\end{lemma}
\begin{IEEEproof}
Since the parity-check matrix is a $(2n-2k) \times (2n)$ matrix, it follows that the code is a $[2n,2k]_q$ code.
The code $\cC$ has covering radius $R$ and hence each vector column of length $n-k$ is a linear combination of at
most $R$ columns of $H$. Hence, given a column vector $(x,y)^t$ (the transpose of $(x,y)$), where $x,y \in \F_q^{n-k}$, we have that $x^t$
is a linear combination of at most $R$ columns of $H'$ and $y^t$ is a linear combination of at most $R$ columns of $H'$.
Therefore, each column vector $(x,y)^t$ is a linear combination of $2R$ columns of $H'$ and hence $\cC'$ is
a $[2n,2k]_q$ linear $(2n,2R)_q$-covering code $\cC'$.

Assume that $\cC$ is a constacyclic code, where $(c_0,c_1,\ldots,c_{n-1})$ and $(d_0,d_1,\ldots,d_{n-1})$ are two codewords
(not necessarily distinct) of $\cC$ and also $(c_1,\ldots,c_{n-1},c_n)$, $(d_1,\ldots,d_{n-1},d_n)$ are codewords of $\cC$.
Consider the words
$(c_0,d_0,c_1,d_1,\ldots,c_{n-1},d_{n-1})$, $(d_0,c_1,d_1,\ldots,c_{n-1},d_{n-1},c_n)$, and $(c_1,d_1,\ldots,c_{n-1},d_{n-1},c_n,d_n)$.
By definition of $\cC$, $\cC'$, $H$, and $H'$, we have that
$H \cdot (c_0,c_1,\ldots,c_{n-1})^t = \bfzero$, $H \cdot (c_1,\ldots,c_{n-1},c_n)^t = \bfzero$,
$H \cdot (d_0,d_1,\ldots,d_{n-1})^t = \bfzero$, and $H \cdot (d_1,\ldots,d_{n-1},d_n)^t = \bfzero$. This implies that
$$
H' \cdot (c_0,d_0,c_1,d_1,\ldots,c_{n-1},d_{n-1})^t = \bfzero~,
$$
$$
H' \cdot (d_0,c_1,d_1,\ldots,c_{n-1},d_{n-1},c_n)^t = \bfzero~,
$$
and
$$
H' \cdot (c_1,d_1,\ldots,c_{n-1},d_{n-1},c_n,d_n)^t = \bfzero~.
$$
Hence, $\cC'$ is a constacyclic code and similarly if $\cC$ is a cyclic code, then $\cC'$ is a cyclic code.
\end{IEEEproof}

%
%

Lemma~\ref{lem:interP} can be implemented for example on $\cH_q(r)$ to obtain $(n,2)_q$-CSs whose length is within a constant
factor of optimality.

\begin{example}
\label{ex:doubHamm}
Let $[ \alpha^0 ~ \alpha ~ \alpha^2 ~ \cdots ~ \alpha^{2^r-2} ]$ be the parity-check matrix of the $[2^r-1,2^r -r-1]_2$ Hamming code $\cH_2(r)$
with radius $1$. We apply Lemma~\ref{lem:interP} and obtain the parity-check matrix
$$
\left[
\begin{array}{ccccccccc}
1 & 0 & \alpha & 0 & \alpha^2 & \cdots & 0 & \alpha^{n-1} & 0 \\
0 & 1 & 0 & \alpha & 0 &  \cdots  & \alpha^{n-2} & 0 & \alpha^{n-1} \\
\end{array}
\right]
$$
is a parity-check matrix of a $[2^{r+1}-2,2^{r+1}-2r-2]_2$ linear $(2^{r+1}-2,2)_2$-covering code.
By applying Construction MC on this code we obtain a code whose number of degenerated words is
the same as the number of codewords in the CSC obtained from $\cH_2(r)$. This number is less than
$2^{2^r-2r-1} +2^{2^{r-1}}$ and hence the related $(2^{r+1}-2,2^{r+1}-2,2)_2$-CSC has at most
$\frac{2^{2^{r+1}-2r-2}}{2^{r+1}-2} + 2^{2^r -2r-1} +2^{2^{r-1}}$ codewords.
The length of the obtained $(2^{r+2}-2,2)_2$-CS by Construction MC is
$$
(\frac{2^{2^{r+1}-2r-2}}{2^{r+1}-2} + 2^{2^r -2r-1} +2^{2^{r-1}}) (2^{r+2}-5)~.
$$
On the other hand, the trivial lower bound on the size of $(2^{r+2}-2,2)$-covering code is $2^{2^{r+1}-2r-3}$.
Therefore, the obtained code is within a factor of $4$ from the lower bound, when in general the smallest
$(2^{r+2}-2,2)_2$-covering code has size $2^{2^{r+1}-2r-2}$ and the obtained code is within a factor of $2$ of this size.

\hfill\quad $\blacksquare $
\end{example}

Similarly to Example~\ref{ex:doubHamm} we can give $(n,2)_q$-CSs whose length is within a small factor of the lower bound
for $(n,2)_q$ when the sequence is over $\F_q$, where $q > 2$. This can be done with the non-binary cyclic codes
and the non-binary constacyclic codes presented in Section~\ref{sec:nonbinary}.

There are a few interesting families of codes obtained by interleaving of parity-check matrices.
Some of these codes are described in the following few examples.
\begin{example}
It is possible to interleave $R$ copies of the parity-check matrix of $\cH_2(r)$.
The outcome is a $(R(2^r-1),R(2^r-1),R)_q$-CSC. This can be generalized
for $\cH_q(r)$ when $\gcd (r,q-1)=1$.
\hfill\quad $\blacksquare $
\end{example}

\begin{example}
Consider the parity-check matrix of $\cH_q(r)$ whose size is $r \times \frac{q^r-1}{q-1}$. It forms
an optimal $(\frac{q^r-1}{q-1},q^r -1,1)_q$-CSC. Interleaving $R$ copies of the parity-check matrix
form a $(\frac{q^r-1}{q-1} R,(q^r -1)R,R)_q$-CSC.

\hfill\quad $\blacksquare $
\end{example}

\begin{example}
The $[23,12]_2$ Golay code is a cyclic code with covering radius $3$. It forms an optimal $(23,23,3)_2$-CSC with $180$ codewords.
Interleaving $k$ copies of its parity-check matrix form a $(23k,23k,3k)_2$-CSC.
\hfill\quad $\blacksquare $
\end{example}

\begin{example}
The $[11,6]_3$ Golay code is a cyclic code with covering radius $2$. It forms an optimal $(11,11,2)_3$-CSC with $69$ codewords.
Interleaving $k$ copies of its parity-check matrix form a $(11k,11k,2k)_3$-CSC.

\hfill\quad $\blacksquare $
\end{example}

\begin{example}
A constacyclic quasi-perfect $[(3^r-1)/2, (3^r-1)/2 -2r]_3$ with covering radius $3$ was presented in~\cite{DDR11}.
It forms a $((3^r-1)/2,3^r-1,3)_3$-CSC.
Interleaving $k$ copies of its parity-check matrix form a $((3^r-1)k/2,(3^r-1)k,3k)_3$-CSC.
\hfill\quad $\blacksquare $
\end{example}

\section{Conclusions and Open Problems for Future Research}
\label{sec:conclude}

We have considered binary and non-binary covering sequences and covering-sequences codes.
Our brief exposition raises several interesting problems for future research.

\begin{enumerate}
\item The same techniques used in the paper can be further used to obtain $(n,R)_q$-CS sequences with larger radii,
but the codes obtained will start to be within a larger optimality factor. Can the method be improved
to obtain sequences of shorter length? Such an improvement or a new construction for shorter sequences is interesting
when $R$ is small and when $R$ is large.

\item The number of cyclic sequences obtained from the codewords in $\cH_2(r)$ can be calculated based
on the formula for the number of necklaces of order $n$ for each length~\cite{Etz24}.
But the formula is quite complicated. Can these computations be simplified for exact computation
of the sizes of the derived $(n,n,1)_q$-CSCs? The situation is even more complicated for
the $(n,q^r -1,1)_q$-CSCs based on the representation of $\cH_q(r)$ as a constacyclic code.

\item The exact computation for the number of codewords in a $(\frac{q^r-1}{q-1} R,(q^r -1)R,R)_q$-CSC is
also of some interest. When $\frac{q^r-1}{q-1}$ is a repunit prime, this computation is relatively easier
and was given in Section~\ref{sec:nonbinary}. But, other computations including for constacyclic codes
with no codewords of period $q^r-1$ are also of considerable interest,

\item In the context of the previous problems, it is interesting to consider the possible periods of the
orbits in an $(n,q^r -1,1)_q$-CSC which depend on the divisors of $q^r -1$, but as was mentioned, when
$\frac{q^r-1}{q-1}$ is not a prime, not for all divisors of $q^r-1$ there are orbits with the associated divisor as a codeword.
Moreover, not all degenerated necklaces are codewords for periods where there are codewords, in contrary to
the case of $\cH_2(r)$ as proved in Lemma~\ref{lem:BnoFull}. This poses many interesting questions.
These and the related questions mentioned before are of enumerative combinatorial nature.

\item More analysis on the constructed $(n,m,R)_q$-CSCs is required. It is also important to find new constructions for
such codes that are not derived directly from cyclic codes or constacyclic codes.
\end{enumerate}





\end{document}